\documentclass[11pt]{article}

\usepackage{graphicx} 
\usepackage{amssymb,amsmath,amsthm,csquotes,mathtools,amssymb}
\usepackage{bbm}
\usepackage[backend=biber,style=authoryear-comp,citestyle=authoryear-comp,maxnames=1]{biblatex}
\usepackage{mathtools}
\usepackage{csquotes}
\usepackage{amsthm}
\usepackage{amsmath}
\usepackage{mathabx}
\usepackage{dsfont}
\usepackage{enumerate}
\usepackage{biblatex}
\usepackage{hyperref}
\usepackage{physics}
\usepackage{multirow,array}
\usepackage{xcolor}
\newtheorem{claim}{Claim}
\usepackage{subcaption}
\usepackage[margin=1in]{geometry}
\usepackage{tikz}
\usetikzlibrary{positioning, arrows.meta}

\newtheorem{definition}{Definition}
\newtheorem{theorem}{Theorem}

\newtheorem{corollary}[theorem]{Corollary}

\title{Quantum Primitive for Output-Hiding Function Sharing: Strategic Settings}

\author{Olivia R. Hartzell\thanks{This work is the subject of U.S. Provisional Patent Application no. 64/045,601.\\
\copyright 2026 Olivia Hartzell. This work is licensed under a Creative Commons Attribution-NonCommercial-NoDerivatives 4.0 International License (CC BY-NC-ND 4.0) available at https://creativecommons.org/licenses/by-nc-nd/4.0/. Any commercial reuse or creation of derivative works without the express written consent of the author is prohibited.}}

\date{\today}

\begin{document}
\maketitle

\begin{abstract}
Applications of the proposed primitive: \textit{Quantum Primitive for Output-Hiding Function Sharing} are discussed for environments in which the parties may have strategic considerations over the generated function value, and therefore may not be mutually trusted. We show that in environments in which the measurement may not be reliably controlled by either party, the primitive permits the parties to generate private, unbiased coins. 
\end{abstract}

\section{Introduction}
While the standard model introduced in \cite{hartzell_quantum} assumes that the parties are mutually trusted in that they are assumed not to collude with any attacker, we extend the model to a setting in which either party may have \textit{preferences} over the generated function $s(x_a, x_b)$. For instance, $s = 1$ may indicate whether Alice becomes the leader in an election. Therefore, an untrusted Alice may attempt to \textit{bias} the outcome in her favor. In what follows, we show that the maximum bias that any cheating party can guarantee on $s$ is zero, when the measurer is permitted to choose any measurement strategy that she likes.

We note that the security properties discussed in \cite{hartzell_quantum} over $s(x_a, x_b)$ hold in this setting, when parties may have preferences over the outcome, but \textit{do not} collude with any attacker or external parties.

\subsubsection{Notable Differences from Two-Party Cryptography}
As discussed in \cite{hartzell_quantum}, the proposed primitive offers a fundamentally different information-sharing structure than standard secure multi-party computation. In the proposed setting,  the joint measurement outcome permits the parties to learn each others' local information simultaneously, while the generated function value $s(x_a, x_b)$ remains information-theoretically hidden from others. By contrast, for parties to simultaneously learn one anothers' local, private information in a classical setting, they require either a trusted, shared authority -- who then learns the parties' local information -- or cryptographic infrastructure. 

Classical cryptographic protocols, including those that rely on commit-then-open procedures, are not truly simultaneous (\cite{ben2002fair}, \cite{BrassardChaumCrepeau1988}, \cite{Naor1991}, \cite{Pedersen1991}, \cite{Ajtai1996}, \cite{KawachiTanakaXagawa2008}).  Instead, cryptographic protocols tend to approximate simultaneity through mechanisms such as commitments or gradual revelation, and require timing bounds, a trusted party, or other physical assumptions to guarantee simultaneity. In practice, when one party can learn the others' information before revealing their own, they may abort the protocol -- and exploit a last-mover advantage-- even if information disclosure is gradual or enforced by time-lock commitments. As discussed in detail in \cite{hartzell_quantum}, any such classical protocol cannot satisfy the information-theoretic properties of the joint measurement or shared function value $s$. 

In general, commitment schemes satisfy two key security properties: \textit{hiding}, ensuring that the committed value remains secret until the reveal phase, and \textit{binding}, ensuring that the committer cannot change the value after committing. When the parties follow the primitive, the information that they choose to share is binding in that once the state has left their local system, they cannot change it. It is hiding in that no party learns the others' information until the \textit{joint} measurement. In contrast with traditional commitment procedures, there is no period where either party could hold evidence of another's hidden commitment, until they receive the others' physical system or the joint measurement is performed. 

This information revelation property is made possible only through joint measurement outcomes, as local measurement outcomes must abide by no-signaling constraints. The measurement outcome needs to be correct for information to be successfully revealed between the parties, however, the measurer need not be trusted -- when the parties abide by the primitive, either she learns $s$ with likelihood no better than random, or she generates an error. 

Overall, this measurement-based procedure offers a more streamlined approach to information revelation relative to \cite{kent1999unconditionally}-style quantum bit commitment protocols which rely on classical transcripts, synchronized clocks, and multiple agents per party to guarantee simultaneous, private information disclosure, while preserving the side-channel attack resistance properties of outsourcing measurements to a third party.

Importantly, however, as established by the Mayers–Lo–Chau no-go theorem (\cite{Mayers1997}, \cite{LoChau1997}), unconditionally secure quantum bit commitment that satisfies both hiding and binding properties is impossible. The proposed primitive clearly does not bypass this impossibility, which in the standard setting assumes that either party who deviates from the primitive can exert complete control over all systems outside their opponent’s local subsystem. Clearly, a cheating Alice who may perform the joint measurement outcome herself, may learn of Bob's input and announce any measurement outcome that she likes.

 I do not attempt to propose any alternative frameworks that contrast with traditional two-party cryptography. I merely show by example, that the maximum bias of $s$ that any unilaterly deviating party can ensure is zero, when the measurer is unrestricted insofar as she is free to implement any measurement strategy that she likes. In this setting, the primitive may be used to generate a private, unbiased coin. The resulting game through outsourcing the joint measurement does not map to a point game as in \cite{kitaev2003quantum_coin}, \cite{mochon2007quantum_weak}. In this case, the measurer is a removed system for the parties' information sharing, but is \textit{not} a trusted party for correctness or secrecy. Under this assumption, the primitive may be used to fairly generate shared, private functions when parties may have strategic incentives, in addition to its applications in secure communications.

When this physical assumption is practically valid, this property, along with simultaneous information-sharing, suggests additional potential applications for quantum resources, particularly for strategic or economic settings where simultaneous revelation is the gold standard and aborts in classical settings are undesirable. Applications may including fair contract signing, elections, or coordination actions in distributed systems. Other important economic applications may include auctions, financial transactions, or matching problems without a central mediator, particularly when parties wish to \textit{obscure} their joint decisions.

\section{Private Coin Flipping}\label{sec:coin_flipping}
In this section, I show that certain constructions of the primitive may be used to generate a \textit{private} unbiased coin between the parties. In particular, when the parties are adversarial in that they have individual \textit{preferences} over the value of the coin, I show that the maximum bias that any unilaterly deviating party can ensure is zero, when the measurer is unrestricted. In some cases, biasing the coin may sacrifice privacy of both the coin and parties' local information. 

Consider the example primitive described in Figure \ref{fig:coin_flipping}. The coin, the recorded value of $s$, is biased if the probability of generating any value is $> \frac{1}{2}$, conditioned on the outcome being shared between the parties (i.e. no error occurs).  Clearly when all parties follow the primitive, $s$ is unbiased, while the security properties described in \cite{hartzell_quantum} are satisfied.

Deviations are defined as follows.
\begin{definition}[Party deviation]
A deviation by a party consists of replacing the actions prescribed by the primitive with an arbitrary strategy as follows.

First, the party may prepare any initial pure state
\[
\ket{\psi} \in \mathcal H_A \otimes \mathcal H_B \otimes \mathcal H_E,
\]
where $\mathcal H_E$ is an arbitrary ancilla system that may remain under the party's control.

Second, the party may apply any local unitary
\[
U_A : \mathcal H_A \rightarrow \mathcal H_A
\]
to their subsystem.

Finally, upon receiving a measurement announcement $j$ from the measurer, the party may record any value $s \in \{0,1\}$ according to an arbitrary strategy, which may differ from the recording rule prescribed by the primitive.
\end{definition}

\begin{definition}[Measurer deviation]
A deviation by the measurer consists of replacing the measurement and announcement rule prescribed by the primitive with an arbitrary strategy acting on the systems received from the parties.

Let $\ket{\phi} \in \mathcal{H}_A \otimes \mathcal{H}_B$ denote the post-unitary state sent to the measurer. The measurer may perform any measurement described by a POVM
\[
\{\Pi_k\}_{k \in \mathcal{K}}, \qquad
\Pi_k \ge 0, \qquad
\sum_{k \in \mathcal{K}} \Pi_k = \mathbb{I}_{\mathcal{H}_A \otimes \mathcal{H}_B}.
\]
which produces an outcome $k \in \mathcal{K}$.

After obtaining the measurement outcome, the measurer may announce any value $j \in \mathcal{J}$ according to an arbitrary (possibly stochastic) map
\[
P(j \mid k)
\]
\end{definition}

It is assumed that the other constraints of the primitive hold in that neither party can gain physical access to the other party's post-unitary subsystem-- this subsystem is obtained by the measurer only. Clearly, when this assumption does not hold, or when one adversarial party may control the third-party measurement, $s$ may not be unbiased.
Additionally, given that we consider adversarial parties, we only consider either party's unilateral deviation, assuming the opposing party abides by the primitive.

\begin{claim}\label{ref:claim_coin1}
    If the measurer abides by the primitive, and one party deviates, then $s$ is unbiased. 
\end{claim}
See \ref{sec:private_coin_flipping_proof} for a proof.

\begin{claim}\label{ref:claim_coin2}
    If the parties follow the primitive, and the measurer deviates, then $s$ is unbiased.
\end{claim}
See \ref{sec:proof_coin_2} for a proof.
Therefore, if $s$ is biased, the measurer and at least one party must deviate from the primitive. Additionally, the following corollary holds.

\begin{corollary}
    The maximum bias in $s$ that any unilaterally deviating cheating party can guarantee is zero, without restricting the measurer's strategy. 
\end{corollary}

 If Alice deviates from the primitive, either she chooses a unitary prescribed by the primitive or she deviates with some other local operation. In either case, by Claim \ref{ref:claim_coin1}, if the measurer follows the primitive, by measuring in $\mathcal{J}$ and announcing the measured outcome, $s$ is unbiased.

Any measurer can diverge from a prescribed deviation with at least one strategy that does not bias $s$, namely by following the primitive, or by measuring in $\mathcal{J}$ and making a measurement announcement that is uniform random over all $j \in \mathcal{J}$. Therefore, for any single party's deviation, the measurer can always choose a strategy that does not bias $s$, and correspondingly, the maximum bias in $s$ that any cheating party can ensure is zero, without restricting the measurer to \textit{not} choose such a strategy. 

When a cheating party's preferences over $s$ are known to the measurer, successfully biasing the coin is trivially fully revealing. The measurer, knowing the value of $s$, and the measurement announcement, additionally learns both parties' shared local information. Further, if the measurer performs a POVM that identifies multiple possible states with positive probability, where at least two states have different corresponding values $\Tilde{s}$, and whose measurement announcements $j$ differ -- any attempt to bias the coin is also trivially fully revealing. In these settings, generating favorable outcomes comes at the \textit{cost of security}.


\begin{figure}[ht]
    \centering
    \[
    \begin{array}{c|cccc}
       & U_b & U_b' & U_b'' & U_b''' \\ \hline
    U_a & (\ket{\phi_{00}} , {0}) & (\ket{\phi_{11}}, {1}) & (\ket{\phi_{10}}, {0}) & (\ket{\phi_{01}}, {1})\\
    U_a'' & (\ket{\phi_{10}}, {1}) & (\ket{\phi_{01}}, {0}) & (\ket{\phi_{00}}, {1}) & (\ket{\phi_{11}}, {0}) \\
    \end{array}
    \]
    \caption{Measurement outcome and recording strategy pairs, given corresponding joint unitaries, where rows correspond to Alice's unitaries and columns correspond to Bob's. Each measurement outcome denoted in an entry occurs with probability one under corresponding joint unitaries.}
    \label{fig:coin_flipping}
\end{figure}

\section{Discussion}
In summary, we have shown that a sufficient condition for fairness is that no deviating party can restrict the measurer's available measurement strategy. Under this condition, the measurer need not be trusted to follow the prescribed protocol and may adopt an arbitrary measurement strategy. Provided the measurer does not collude with either party, the primitive nevertheless generates an unbiased and private realization of the shared value $s$.

These results suggest that the primitive and the security properties introduced in \cite{hartzell_quantum} are useful not only for secure communication and distributed computation, but also for decentralized game-theoretic settings in which parties require a common random outcome without relying on a trusted mediator. Rather than assuming an honest referee, fairness follows from the physics of the generated quantum states along with the recording strategy. More broadly, this illustrates how quantum resources can be used to implement agreement protocols in which the outcome is jointly generated while remaining information-theoretically hidden from the physical system that mediates the interaction.

Several directions remain for future work. An important open question is to characterize the full class of functions whose bias is invariant to parties' deviations. We leave the more general question of characterizing the class of mechanisms implementable in such quantum environments without trusted mediators for future work. 

\newpage

\newpage

\printbibliography

\newpage

\section{Appendix}
\subsection{Proof of Claim \ref{ref:claim_coin1}}\label{sec:private_coin_flipping_proof}
\begin{proof}
Suppose that the measurer abides by the primitive and measures the received state $\ket{\phi} \in \mathcal{H}_A \otimes \mathcal{H}_B$ in $\mathcal{J}$, and announces the measured outcome $j \in \mathcal{J}$.  

    Consider a cheating Alice who wishes to bias the outcome of the value of $s \in \{0, 1\}$ so that either value is generated with probability $> \frac{1}{2}$. 
    
     For any arbitrary $\ket{\psi} \in \mathcal{H}_A \otimes \mathcal{H}_B$ that the parties can act on, there is no unitary that Alice can apply to her local subsystem that biases the value of $s$ when the resulting state is measured in $\mathcal{J}$ and the announcement corresponds with the measurement, assuming that Bob randomizes uniformly over his feasible set of unitaries $\{U_b, U_b', U_b'', U_b'''\}$.  

    This is due to the fact that, for any $\ket{\psi} \in \mathcal{H}_A \otimes \mathcal{H}_B$ and \textit{any} valid $U_A$, which may deviate from those specified by the primitive,

\[
\begin{aligned}
&\mathrm{Tr}\Big[(U_A \otimes U_b) \ket{\psi}\bra{\psi} (U_A^\dagger \otimes U_b^\dagger) \, \pi_{00}^{\mathcal{J}}\Big] 
= \mathrm{Tr}\Big[(U_A \otimes U_b') \ket{\psi}\bra{\psi} (U_A^\dagger \otimes U_b'^\dagger) \, \pi_{11}^{\mathcal{J}}\Big] \\
&= \mathrm{Tr}\Big[(U_A \otimes U_b'') \ket{\psi}\bra{\psi} (U_A^\dagger \otimes U_b''^\dagger) \, \pi_{10}^{\mathcal{J}}\Big] 
= \mathrm{Tr}\Big[(U_A \otimes U_b''') \ket{\psi}\bra{\psi} (U_A^\dagger \otimes U_b'''^\dagger) \, \pi_{01}^{\mathcal{J}}\Big],
\end{aligned}
\]

and

\[
\begin{aligned}
&\mathrm{Tr}\Big[(U_A \otimes U_b) \ket{\psi}\bra{\psi} (U_A^\dagger \otimes U_b^\dagger) \, \pi_{10}^{\mathcal{J}}\Big] 
= \mathrm{Tr}\Big[(U_A \otimes U_b') \ket{\psi}\bra{\psi} (U_A^\dagger \otimes U_b'^\dagger) \, \pi_{01}^{\mathcal{J}}\Big] \\
&= \mathrm{Tr}\Big[(U_A \otimes U_b'') \ket{\psi}\bra{\psi} (U_A^\dagger \otimes U_b''^\dagger) \, \pi_{00}^{\mathcal{J}}\Big] 
= \mathrm{Tr}\Big[(U_A \otimes U_b''') \ket{\psi}\bra{\psi} (U_A^\dagger \otimes U_b'''^\dagger) \, \pi_{11}^{\mathcal{J}}\Big].
\end{aligned}
\]

where $\pi_{ij}^{\mathcal{J}}$ denotes the operator corresponding with measurement outcome $\ket{\phi_{ij}} \in \mathcal{J}$. Meaning, given the recording strategy, for any arbitrary unitary $U_A$ that Alice chooses, for every unitary of Bob's under which he records a value of $s = 0$, there is another equally likely unitary of Bob's that results in him recording $s = 1$ and vice-versa. The same holds for any mixed-state preparations. If Alice's recorded value of $s$ does not correspond with Bob's, she generates an error and does not bias $s$.

In order to bias $s$ from a feasible deviation, Alice must therefore learn information about Bob's chosen unitary. 
Even if Alice prepares an initial state $\ket{\psi} \in \mathcal{H}_A \otimes \mathcal{H}_B \otimes \mathcal{H}_E$ entangled with an ancilla $\mathcal{H}_E$, since Bob's unitary $U_B$ acts only on $\mathcal{H}_B$, the reduced density matrix on Alice's systems,
\[
\rho_{AE} = \mathrm{Tr}_B\big[(I_A \otimes U_B \otimes I_E)\, \ket{\psi}\bra{\psi}\, (I_A \otimes U_B^\dagger \otimes I_E)\big],
\]
is independent of $U_B$.
Without access to Bob's post-unitary system, under her feasible deviations, Alice cannot bias $s$. 
 The analogous argument holds for Bob.  
\end{proof}

\subsection{Proof of Claim \ref{ref:claim_coin2}}\label{sec:proof_coin_2}
\begin{proof}
    This follows from the recording strategy given by the primitive in Figure \ref{fig:coin_flipping}.
    Suppose the parties follow the primitive by randomizing uniformly over their respective unitary sets given by the primitive, acting on any valid initial state $\ket{\psi} \in \mathcal{H}_A \otimes \mathcal{H}_B$, and recording values of $s$ according to the recording strategy prescribed by the primitive. 

    Let the parties' chosen joint unitary be denoted by $U$. The measurer can perform any deviation as previously defined. For any such feasible deviation, either she announces the correct state $j = j^*(U)$ or $j \neq j^*(U)$. 
    
    By the properties of the recording strategy, any incorrect announcement $j \neq j^*(U)$ results in an error. Therefore, the probability of any value of $s$ that is not generated in error, is the probability that $j = j^*(U)$, times the probability of the parties choosing a joint unitary $U$ that correspond with the value of $s$. Given that when the parties follow the primitive and sample $U$ uniformly from $\mathcal{U}$, and $|\{U \in \mathcal{U} : f^s(U)=1\}| = |\{U \in \mathcal{U} : f^s(U)=0\}|$, where $f^s$ denotes the restriction of $f(U)$ to its $s$ component, the resulting value $s$ is unbiased.

\end{proof}

\end{document}